\documentclass[twoside,12pt]{article}
\usepackage{amsmath,amssymb,amsthm,amsfonts,mathrsfs,wasysym,latexsym,times,lineno, subfigure,color,booktabs}
\usepackage{amsmath,amsfonts}
\usepackage{multirow}
\usepackage{array}
\usepackage{amsthm}
\usepackage{graphicx}
\usepackage{color}
\usepackage{multicol}
\usepackage{float}
\usepackage{cite}
\usepackage[title]{appendix}

\newtheorem{thm}{Theorem}[section]

\newtheorem{definition}{Definition}[section]

\newtheorem{lemma}[thm]{Lemma}

\newtheorem{theorem}[thm]{Theorem}

\newtheorem{corollary}[thm]{Corollary}

  \date{}
\title{Self-referential instances of the dominating set problem are irreducible}
\author{ Guangyan Zhou\\
\footnotesize Department of Mathematics and Statistics, Beijing Technology and Business University, Beijing, 100048, China\\
\footnotesize zhouguangyan@btbu.edu.cn
}
\begin{document}
\maketitle

\begin{abstract}

We study the algorithmic decidability of the domination number in the Erdős–Rényi random graph model $G(n,p)$. We show that for a carefully chosen edge probability $p=p(n)$, the domination problem exhibits a strong irreducible property. Specifically, for any constant $0<c<1$, no algorithm that inspects only an induced subgraph of order at most $n^c$ can determine whether $G(n,p)$ contains a dominating set of size $k=\ln n$. We demonstrate that the existence of such a dominating set  can be flipped by a local symmetry mapping that alters only a constant number of edges, thereby producing indistinguishable random graph instances which requires exhaustive search. These results demonstrate that the extreme hardness of the dominating set problem in random graphs cannot be attributed to local structure, but instead arises from the self-referential nature and near-independence structure of the entire solution space.
\end{abstract}

\section{Introduction}
The domination set problem is a classic NP-hard problem in graph theory. Given an Erd\H{o}s-R\'{e}nyi random graph $G(n,p)$, that is, a graph on $n$ vertices in which each of the $\binom n2$ potential edges appears independently with probability $p$. A dominating set is a subset of vertices such that every vertex of $G(n,p)$ either belongs to the set or is adjacent to at least one vertex in the set. Despite extensive study, no algorithm substantially better than the trivial brute-force search, which checks all subsets of size $k$, is known in general. This difficulty  has motivated research into alternative approaches, including fixed-parameter tractable (FPT) algorithms and approximation schemes \cite{downey95,chen2006}.

From the perspective of parameterized complexity, however, dominating set is known to be W[2]-hard, which strongly suggests that no algorithm with running time $f(k)n^{O(1)}$ exists unless the $W$- hierarchy collapses. In parallel, significant progress has been made in understanding the domination number of the random graphs. In particular, Wieland and Godbole  \cite{wieland2001} showed that for certain ranges of $p$, mirroring classical concentration phenomena observed for other graph parameters such as the independence number and the chromatic number. In contrast, when $p$ is small, for example $p=O(1/n)$, the domination number exhibits markedly different behavior: it is not concentrated on any interval of length $o(\sqrt n)$, and the threshold for domination shifts from the first-moment bound to the median, with a substantial gap between the two \cite{Glebov}. 

In this paper, we establish a strong indistinguishability phenomenon for the domination set problem in random graphs. Specifically, for any constant $0<c<1$, we show that it is impossible to determine whether $G(n,p)$ contains a dominating set of size $k=\ln n$  by inspecting any induced subgraph $ H$ of order at most $n^c$. To prove this result, we carefully choose the edge probability $p$ and show that there exist two types of instances of $G(n,p)$: one class consists of graphs with a unique dominating set of size $k$, while the other consists of graphs that contain no dominating set of size $k$, but instead contain vertex sets of  size $k$ that dominate all but a single vertex. 

A crucial feature of our construction is that the existence or non-existence of a dominating set of size $k$ can be toggled by altering the adjacency between four vertices. This change alters the existence of a dominating set of size $k$ while leaving the subgraph $H$ unchanged. Consequently, any algorithm whose decision is based solely on the information contained in $H$ cannot reliably distinguish between these two types of instances, and therefore cannot determine the existence of a dominating set of size $k$ without effectively examining the entire graph.

Our methodology is closely related to the symmetry-mapping technique introduced in the
construction of self-referential (most indistinguishable) instances for Model RB in \cite{xu2025}.  These instances were shown to require essentially brute-force computation for their solution. The present work extends this line of reasoning to the Dominating Set problem, providing further evidence that partial information cannot substitute for global structure. In particular, our results support the view that the fundamental source of extreme hardness in such problems lies in the self-referential nature and near-independence structure of their candidate solutions.

\section{Main results}

\begin{definition}
Let $G=G(n,p)$ be a random graph. We say that the \emph{dominating set problem} for $G$ is \emph{reducible} if there exists an induced subgraph $H$ of order $n^c$, for some constant $0<c<1$, such that $H$ contains a dominating set of size $k$ if and only if $G$ contains a dominating set of size $k$; otherwise $G$ is  said to be \emph{irreducible}. 
\end{definition}

The following theorem shows that, under this definition, the dominating set problem in random graphs cannot be reduced to any subgraph of order $n^c$ for any constant $0<c<1$.

\begin{theorem}\label{th:main}
The dominating set problem for $G(n,p)$ is irreducible.
\end{theorem}

\subsection{Existence of a unique dominating set of size $k$}
In this section we show that, with positive probability, the random graph $G(n,p)$ contains a unique dominating sets of size $k=\ln n$. Let $X$ be the number of dominating sets of size $k$ in $G(n,p)$. A straightforward calculation yields

\begin{align}
\label{eq:first}\mathbf{E}[X]&=\binom{n}{k}(1-(1-p)^k)^{n-k},\\
\mathbf{E}[X^2]&=\sum_{i=0}^k\binom{n}{i}\binom{n-i}{k-i}\binom{n-k}{k-i}(1-(1-p)^k)^{2(k-i)}(1-2(1-p)^k+(1-p)^{2k-i})^{n-2k+i}.
\end{align}

In the following, we tacitly choose the edge probability $p=p(n)$ so that $E[X]=\delta+o(1)$ for a constant $0<\delta<1$. Next, we try to get the asymptotic expression of $p$. Note that
$$\ln\binom nk+(n-k)\ln(1-(1-p)^k)=\ln\delta,$$
thus
$$(1-p)^k\approx1-\exp\left\{-\frac{\ln^2n}{n-k}\right\}\approx\frac{\ln^2n}n.$$
Thus we have
$$p=1-\frac1e\left[(1+o(1))\ln^2n\right]^{1/\ln n}.$$
Moreover, it is easy to see that the $\lim_{n\to\infty}p=1-1/e$.

\begin{lemma}\label{lem:onethird}
In  $G(n,p)$,
$$\frac{\delta}{\delta+1}\le \mathbf{Pr}(X>0)\le\delta.$$
\end{lemma}

\begin{proof}
The upper bound follows immediately from Markov’s inequality:
\begin{equation}
\mathbf{Pr}(X>0)\le \mathbf{E}[X]=\delta+o(1).
\end{equation}
For a lower bound we apply the second moment method.  $\mathbf{E}[X^2]$ counts ordered pairs of dominating sets of size $k$. Let $F(i)$ denote the contribution from pairs whose intersection has size $i$, then
\begin{align*}
\mathbf{E}[X^2]&=\sum_{i=0}^k F(i),
\end{align*}
where
 \begin{align*}
F(i)=\binom{n}{i}\binom{n-i}{k-i}\binom{n-k}{k-i}\big(1-(1-p)^k\big)^{2(k-i)}\left(1-2(1-p)^k+(1-p)^{2k-i}\right)^{n-2k+i}.
\end{align*}

 If $i=0$, we have
  \begin{align*}
F(0)=\binom{n}{k}\binom{n-k}{k}\big(1-(1-p)^k\big)^{2n-2k}=(1+o(1))\mathbf{E}[X]^2.
\end{align*}
If $i=k$, then 
  \begin{align*}
F(k)=\mathbf{E}[X].
\end{align*}
For $1\le i\le k-1$, 

  \begin{align*}
\frac{F(i)}{\mathbf{E}^2[X]}=\frac{\binom{n}{i}\binom{n-i}{k-i}\binom{n-k}{k-i}}{\binom nk ^2}\left[\frac{(1-2(1-p)^k+(1-p)^{2k-i})}{(1-(1-p)^k)^2}\right]^{n-2k+i}.
\end{align*}

Note that $k=\ln n$, we apply the following asymptotic estimates for $1\le i\le k-1$: $(1-p)^k=(1+o(1))\frac{\ln^2n}n$, and $(1-p)^{2k-i}=(1+o(1))\left(\frac{\ln^2n}n\right)^{1-\frac ik}$. This gives 

\begin{align*}
\frac{F(i)}{\mathbf{E}^2[X]}
\le (1+o(1))\frac{k^{2i}}{n^i}\exp\left\{\frac{(\ln^2n)^{2-\frac ik}}{n^{1-\frac ik}}\right\}.
\end{align*}

Thus, 
\begin{equation}\label{eq:cauchy}
\frac{\mathbf{E}[X^2]}{\mathbf{E}[X]^2}=\frac{\sum_{i=0}^kF(i)}{\mathbf{E}[X]^2}\le1+\frac1{\mathbf{E}[X]}+o(1)=1+\frac1{\delta}+o(1).
\end{equation}

Finally, by the Cauchy-Schwarz inequality,
$$\mathbf{Pr}(X>0)\ge\mathbf{E}[X]^2/\mathbf{E}{[X^2]}\ge\frac{\delta}{\delta+1}+o(1),$$
which completes the proof.
\end{proof}
As an immediate consequence of Lemma \ref{lem:onethird} we obtain a lower bound of the probability that $G(n,p)$ has a unique dominating set of size $\ln n$.
 \begin{corollary}\label{lem:solutionpair}
The probability that $G(n,p)$ has a unique dominating set of size $\ln n$ is at least $\delta(1-\delta)/(1+\delta)$.
\end{corollary}
\begin{proof}
Let $\rho_1$ be the probability that $G(n,p)$ has a unique dominating set of size $\ln n$, and $\rho_{\ge2}$ be the probability that $G(n,p)$ has at least two such dominating sets. Then from Lemma \ref{lem:onethird}, we have
\begin{align*}
&\ \mathbf{E}[X]=\delta\ge\rho_1+2\rho_{\ge2},\\
&\mathbf{Pr}(X>0)=\rho_1+\rho_{\ge2}\ge\frac{\delta}{\delta+1}.
\end{align*}
Therefore $\rho_1\ge\delta(1-\delta)/(1+\delta)$.
\end{proof}
\subsection{Non-existence of  dominating sets of size $k$}
\begin{lemma}\label{lem:noset}
In the random graph $G(n,p)$, if there does not exist a dominating set of size $k=\ln n$, then with high probability there exist $k$-vertex sets that dominates all but exactly one vertex.
\end{lemma}

\begin{proof}
Let $S$ be a $k$-vertex set. We say $S$ dominates all but one vertex if exactly one vertex in $V\setminus S$ has no neighbor in $S$, while the remaining $n-k-1$ vertices in $V\setminus S$ are adjacent to  at least one vertex of $S$.

Let $N=\sum_{|S|=k}I_S$, where $I_S=\mathbf{1}_{\{\text{$S$ dominates all but one vertex}\}}$. A direct calculation gives
\begin{align}\label{eq:first2}
\mathbf{E}[N]=\binom{n}{k}(n-k)(1-p)^k\big(1-(1-p)^k\big)^{n-k-1}.
\end{align}

We compare $\mathbf{E}[N]$ with $\mathbf{E}[X]$. Since
\[
\frac{\mathbf{E}[N]}{\mathbf{E}[X]}
=(n-k)(1-p)^k\bigl(1-(1-p)^k\bigr)^{-1},
\]
and under our choice
\[
(1-p)^k=(1+o(1))\frac{\ln^2 n}{n},
\]
then
\[
\mathbf{E}[N]=(1+o(1))(\ln n)^2\mathbf{E}[X]\to\infty.
\]

\medskip
We now apply the second moment method. Let $S,S'$ be two  $k$-vertex set, and  $i=|S\cap S'|$. Then
\[
\mathbf{E}[N^2]
=\sum_{|S|=k}\sum_{|S'|=k}\mathbf{E}[I_S I_{S'}]=\sum_{i=0}^k \Phi(i)W(i),
\]
where
\begin{equation*}
\Phi(i)=\binom{n}{i}\binom{n-i}{k-i}\binom{n-k}{k-i},
W(i)=\mathbf{E}\left[I_S I_{S'}\bigm| |S\cap S'|=i\right]
= P_1(i)+2P_2(i)+P_3(i)+P_4(i),
\end{equation*}
and (the derivation of $W(i)$ is in the Appendix)
\begin{align*}
P_1(i)&=(k-i)^2\mu^2(1-\mu)^{2k-2i-2}\zeta^m,\\
P_2(i)&=(k-i)m\mu^2\vartheta(1-\mu)^{2k-2i-1}\zeta^{m-1},\\
P_3(i)&=m\varpi(1-\mu)^{2(k-i)}\zeta^{m-1},\\
P_4(i)&=m(m-1)(\mu \vartheta)^2(1-\mu)^{2k-2i}\zeta^{m-2},
\end{align*}
with
$\mu\triangleq(1-p)^k,\vartheta\triangleq1-(1-p)^{k-i},\varpi\triangleq(1-p)^{2k-i},m\triangleq n-2k+i,\zeta\triangleq1-2\mu+\varpi.$

Now we can rewrite (\ref{eq:first2}) as
\[ 
\mathbf{E}[N]=(n-k)\binom nk\mu(1-\mu)^{n-k-1}.
\]

We estimate $\mathbf{E}[N^2]/\mathbf{E}[N]^2$ by considering three cases.

\medskip
\noindent\textbf{Case 1}: $i=0$.
Here
\begin{align*}
\Phi(0)&=\binom{n}{k}\binom{n-k}{k},\\
W(0)&=\big(k^2+2k(n-2k)(1-\mu)^{-1}+(n-2k)+(n-2k)(n-2k-1)\big)\mu^2(1-\mu)^{2n-2k-2}.
\end{align*}
A direct computation yields
\begin{align*}
\frac{W(0)}{(n-k)^2\mu^{2}(1-\mu)^{2n-2k-2}}&=\frac{k^2+2k(n-2k)(1-\mu)^{-1}+(n-2k)+(n-2k)(n-2k-1)}{(n-k)^2}\\
&=1+o(1),
\end{align*}
hence
\[
\frac{\Phi(0)W(0)}{\mathbf{E}[N]^2}=1+o(1).
\]

\medskip
\noindent\textbf{Case 2}: $i=k$.
In this case $S=S'$, so 
\begin{align*}
\Phi(k)=\binom{n}{k},\quad
W(k)=(n-k)\mu(1-\mu)^{n-k-1},
\end{align*}
thus $\Phi(k)W(k)=\mathbf{E}[N]$.
Therefore,
\[
\frac{\Phi(k)W(k)}{\mathbf{E}[N]^2}
=\frac{1}{\mathbf{E}[N]}
=o(1).
\]

\medskip
\noindent\textbf{Case 3}: $1\le i\le k-1$.
It is easy to see that
\begin{align}\label{eq:Phi}
\frac{\Phi(i)}{\binom{n}{k}^2}\le\frac{k^{2i}}{n^i}.
\end{align}
Moreover, note that $k=\ln n$, standard asymptotic estimates give
\[
\mu=(1+o(1))\frac{\ln^2n}n,\qquad
\vartheta=1-(1-p)^{k-i}=1-o(1),
\]
and
\[
\zeta^m=(1-2(1-p)^k+(1-p)^{2k-i})^m\le(1-\mu)^{2m}\exp\{(1+o(1))\ln^2n\}.
\]

\medskip

Therefore,
\begin{align*}
\frac{P_1(i)}{(n-k)^2\mu^2(1-\mu)^{2n-2k-2}}&=\frac{(k-i)^2}{(n-k)^2}(1-\mu)^{-2m}\zeta^{m}=o(1).\\
\frac{P_2(i)}{(n-k)^2\mu^2(1-\mu)^{2n-2k-2}}&=\frac{2(k-i)m}{(n-k)^2}\vartheta(1-\mu)^{-2m-1}\zeta^{m-1}=o(1).\\
\frac{P_3(i)}{(n-k)^2\mu^2(1-\mu)^{2n-2k-2}}&=\frac{m}{(n-k)^2}(1-p)^{-i}(1-\mu)^{-2m}\zeta^{m-1}=o(1).\\
\frac{P_4(i)}{(n-k)^2\mu^2(1-\mu)^{2n-2k-2}}&=\frac{m(m-1)}{(n-k)^2}\vartheta^2(1-\mu)^{-2m}\zeta^{m-2}.
\end{align*}
Note that $\frac{m(m-1)}{(n-k)^2}=1+o(1)$, $\vartheta^2\le1$, and 
$$(1-\mu)^{-2m}\zeta^{m-2}=\exp\left\{(1+o(1))\left(\frac{\ln^2n}{n}\right)^{1-\frac ik}\ln^2n\right\}.$$

By (\ref{eq:Phi}), we have
\begin{align*}
\frac{\Phi(i)}{\binom nk^2}\frac{P_4(i)}{(n-k)^2\mu^2(1-\mu)^{2n-2k-2}}=(1+o(1))k^{2i}\exp\left\{\left(\frac{\ln^2n}{n}\right)^{1-\frac ik}\ln^2n-i\ln n\right\},
\end{align*}
and attains its maximum at around $i=k-1$, thus
\begin{align*}
\sum_{i=1}^{k-1}\frac{\Phi(i)W(i)}{\mathbf{E}^2[N]}=o(1).
\end{align*}

Therefore,
\begin{align*}
\frac{\mathbf{E}[N^2]}{\mathbf{E}^2[N]}\le1+o(1).
\end{align*}
The second moment method yields 
\begin{align*}
\mathbf{P}[N>0]\ge\frac{\mathbf{E}^2[N]}{\mathbf{E}[N^2]}\ge1-o(1).
\end{align*}
Thus with high probability there exists a  $k$-vertex set that dominates all but exactly one vertex, completing the proof. 
\end{proof}

\section{Proof of Theorem  \ref{th:main}}

Let $\mathscr{G}$ denote the family of instances of $G(n,p)$ such that each instance either has a unique dominating set of size $k=\ln n$ or has no dominating set of that size.  By Lemma~\ref{lem:onethird} and Lemma~\ref{lem:noset}, the probability
that a random graph $G(n,p)$ belongs to $\mathscr{G}$ is bounded away from zero.

We show that for any constant $0<c<1$, no subgraph of order at most $n^c$ suffices
to decide whether $G$ contains a dominating set of size $k$.

\medskip
\noindent\textbf{Case 1: $G$ has a unique dominating set of size $k$.}

Assume that $G\in\mathscr{G}$ contains a unique dominating set $S$ of size $k$.
By Corollary~\ref{lem:solutionpair}, this occurs with probability at least $\delta(1-\delta)/(1+\delta)$.
Let $H$ be an arbitrary induced subgraph of $G$ on at most $n^c$ vertices, and let $V_G,V_H$ be the vertex set of $G,H$ respectively.

Then
\[
\mathbf{P}(S\subseteq V_G\setminus V_H)
=\frac{\binom{n-n^c}{k}}{\binom{n}{k}}
=1-o(1),
\]
so with high probability the vertices of $S$ lie entirely outside $H$.

We first show that with high probability there exists a vertex in
$V_G\setminus(V_H\cup S)$ adjacent to exactly one vertex of $S$.
For a vertex $v\notin S$, let
\[
A_v=\{\text{$v$ is adjacent to exactly one vertex of $S$}\},
\quad
B_v=\{\text{$v$ is adjacent to at least one vertex of $S$}\}.
\]
Then
\[
\mathbf{P}(A_v\mid B_v)
=\frac{\mathbf{P}(A_v)}{\mathbf{P}(B_v)}
=\frac{k p(1-p)^{k-1}}{1-(1-p)^k}.
\]
Consequently,
\[
\mathbf{P}\!\left(\exists\,v\in V_G\setminus(V_H\cup S): A_v\right)
=1-\left(1-\frac{k p(1-p)^{k-1}}{1-(1-p)^k}\right)^{n-n^c-k}
=1-o(1).
\]

Fix such a vertex $v$, and let $u\in S$ be its unique neighbor in $S$.
Next, choose  vertices $z\in S$ and
$w\in V_G\setminus(V_H\cup S)$  such that, among these vertices, the only existing edges are $(u,v)$ and $(z,w)$.

We now perform a \emph{symmetry mapping} by removing the edges $(u,v)$ and $(z,w)$, and adding the edges $(u,z)$ and
$(v,w)$, as illustrated in Figure \ref{map} (transition from (a) to (b)). This operation modifies only edges incident to vertices outside $H$, while preserving the degree of every vertex and keeping the total number of edges unchanged. As a result of this transformation, vertex $v$ no longer has any neighbor in $S$, implying that $S$ is no longer a dominating set.

Moreover, with high probability no new dominating set of size $k$ is created.
Indeed, the probability that  $\{v\}$ or $\{w\}$ extends to a
dominating set of size $k$ is at most
\[
\binom{n-n^c}{k-1}\bigl(1-(1-p)^k\bigr)^{n-n^c-k-1}
=o(1).
\]

Thus, with high probability, the modified graph contains no dominating set of
size $k$.

\medskip
\begin{figure}[h]
  \centering
  \includegraphics[width=0.8\columnwidth]{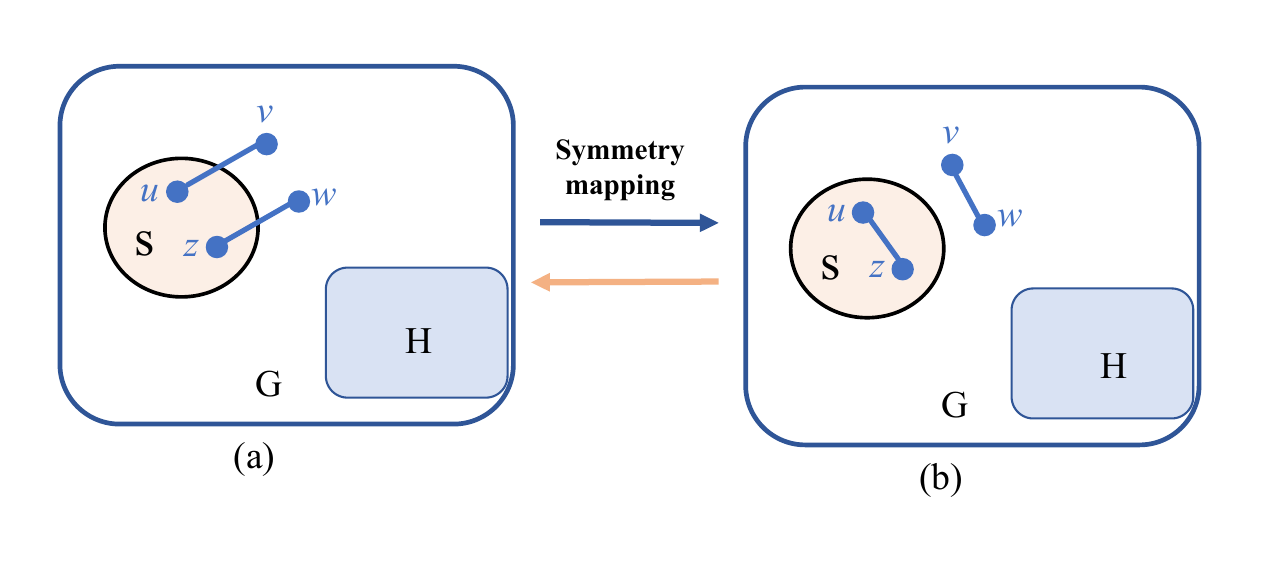}
\caption{\footnotesize{
A symmetry mapping between two classes of instances. \\
  (a) → (b) Initially, $G$ has a unique dominating set of size $k$, and the vertex $v$ has exactly one neighbor $u\in S$. The symmetry mapping transforms $G$  into a graph with no dominating set of size $k$.\\
(b) → (a) Initially, $G$ has no dominating set of size $k$, and $v$ is the only vertex not dominated by $S$. The symmetry mapping transforms $G$  into a graph with a unique dominating set of size $k$.
}\label{map}}
\end{figure}

\noindent\textbf{Case 2: $G$ has no dominating set of size $k$.}

Now suppose that $G\in\mathscr{G}$ has no dominating set of size $k$.
By Lemma~\ref{lem:noset}, with high probability there exists a set
$S$ of $k$ vertices and a vertex $v$ such that $S$ dominates all vertices except
$v$.

Let $H$ be any induced subgraph of order at most $n^c$. Then
\[
\mathbf{P}(S\cup\{v\}\subseteq G\setminus H)
=\frac{\binom{n-n^c}{k+1}}{\binom{n}{k+1}}
=1-o(1),
\]
so with high probability all vertices in $S\cup\{v\}$ lie outside $H$.

Choose vertices $u,z\in S$ and $w\in V_G\setminus(V_H\cup S)$ such that, among these vertices, the only existing edges are $(v,w)$ and $(u,z)$. Apply the symmetry mapping in the opposite direction by
removing the edges $(v,w)$ and $(u,z)$ and adding the edges $(v,u)$ and $(z,w)$. After this change, the vertex $v$ becomes
adjacent to $S$, and hence $S$ forms a dominating set of size $k$. See Figure \ref{map}(from (b) to (a)).

As before, with high probability no other dominating set of size $k$ is created.
Therefore, with high probability, the modified graph contains a unique dominating
set of size $k$.

\medskip

In both cases, by altering only edges whose endpoints lie outside $H$, we can
flip the existence of a dominating set of size $k$ while keeping the induced
subgraph $H$ unchanged. Consequently, for any $0<c<1$, no subgraph of order
$n^c$ contains sufficient information to determine whether $G(n,p)$ has a
dominating set of size $k$.

This proves that the dominating set problem for $G(n,p)$ is irreducible, and
completes the proof of Theorem~\ref{th:main}.

\section{Conclusion}
Our results provide another example of self-referential instances that inherently require exhaustive search, in the sense that no proper substructure contains sufficient information to decide the global property. In such instances, the existence of solutions cannot be inferred from any sublinear portion of the input, reflecting a fundamental obstruction to local or reduction-based algorithms. We argue that this phenomenon is closely related to the presence of independent solution spaces in the underlying combinatorial problem, where candidate solutions interact only weakly and local modifications can drastically alter global feasibility.

This mechanism naturally arises in certain constraint satisfaction problems. In the dominating set problem studied here, for example, two randomly chosen candidate vertex sets of size $n^c$ are disjoint with high probability, implying that their domination properties are nearly independent. This independence enables the construction of symmetry mappings that preserve all local subgraphs while flipping the existence of global solutions, thereby producing  most indistinguishable instances. Similar phenomena have been observed in other classical problems. Notably, for the CLIQUE problem, Li et al.~\cite{Li2025} constructed self-referential instances that expose an intrinsic source of algorithmic hardness beyond local graph structure.

Taken together, these results suggest that by understanding how self-reference and solution-space independence forces exhaustive search in specific cases, one can gain insight into the nature of extreme hardness, making ``proving computational hardness not hard in mathematics" for these particular examples. This aligns with Gödel's use of self-reference in logic to prove the existence of formally unprovable mathematical statements.

\begin{appendices}
\section{Derivation of the function $W(i)$}
 Let $S,S'$ be two sets, each of size $k$, that each dominate all but one vertex in $G(n,p)$. Suppose that $|S\cap S'|=i$. We derive an explicit expression for $$W(i)=\mathbf{P}\big(I_S=1,I_{S'}=1\bigm| |S\cap S'|=i\big),$$ where $I_S=1$ (respectively $I_{S'}=1)$ denotes the event that $S$ (resp. $S'$) dominates all but one vertex.

Define the disjoint vertex sets 
\begin{align*}
A=S\cap S',B=S\backslash S',C=S'\backslash S,D=V\backslash(S\cup S').
\end{align*}
Then
\begin{align*}
|A|=i,|B|=|C|=k-i,|D|= n-2k+i.
\end{align*}
The events $I_S=1$ and  $I_{S'}=1$ can equivalently restated as follows:
\begin{itemize}
\item $I_S=1$: Among the vertices in $C\cup D$, exactly one vertex has no neighbor  in $S$;
\item $I_{S'}=1$: Among the vertices in $B\cup D$, exactly one has no neighbor in $S'$.
\end{itemize}

Let $x\in C\cup D$ be  the unique vertex not dominated by $S$, and $y\in B\cup D$ be the unique vertex not dominated by $S'$. We consider all possible locations of the pair $(x,y)$.

For convenience, introduce the notation
\begin{align*}
\mu\triangleq(1-p)^k,\vartheta\triangleq1-(1-p)^{k-i},\varpi\triangleq(1-p)^{2k-i},m\triangleq n-2k+i,\zeta\triangleq1-2\mu+\varpi.
\end{align*}

\textbf{Case 1: $x\in C$ and $y\in B$}

The vertex $x$ has no neighbor in $S$ with probability $\mu$, and $y$ has no neighbor in $S'$ with probability $\mu$. Each of the remaining $k-i-1$ vertices in $C$ must be dominated by $S$, and each of the remaining $k-i-1$ vertices in $B$ must be dominated by $S'$, each  occurring with probability $1-\mu$. Each vertex in $D$ must be dominated by both $S$ and $S'$, which happens with probability $\zeta$. Since there are $k-i$ choices for $x$in $C$ and  $k-i$ choices for $y$ in $B$, we obtain 

$$P_1(i)=(k-i)^2\mu^2(1-\mu)^{2k-2i-2}\zeta^m.$$

\textbf{Case 2: $x\in C,y\in D$} (the symmetric case $x\in D,y\in B$ yields the same contribution).

Here $x$ has no neighbor in $S$, which occurs with probability $\mu$. The vertex $y$ has no neighbor in $S'$ but is dominated by $S$, so it has no edges to $A\cup C$ and at least one edge to $B$; this occurs with probability $\mu \vartheta$.  The remaining $k-i-1$ vertices in $C$ are dominated by $S$, and all $k-i$ vertices in $B$ are dominated by $S'$, each with probability $1-\mu$. The remaining $m-1$ vertices in $D$ are dominated by both $S$ and $S'$, each with probability $\zeta$. There are $k-i$ choices for $x$ and $m$ choices for $y$, giving

$$P_2(i)=(k-i)m\mu^2\vartheta(1-\mu)^{2k-2i-1}\zeta^{m-1}.$$

\textbf{Case 3: $x\in D,y\in D$, and $x=y$}

In this case the common vertex $x=y$ has no neighbors in either $S$ or $S'$, which occurs with probability $\varpi$. The remaining $m-1$ vertices in $D$ are dominated by both $S$ and $S'$, each with probability $q$. All $k-i$ vertices in $C$ are dominated by $S$, and  all $k-i$ vertices in $B$ are dominated by $S'$,  each with probability $1-\mu$. Since there are $m$ choices for the common vertex, we obtain

$$P_3(i)=m\varpi(1-\mu)^{2(k-i)}\zeta^{m-1}.$$

\textbf{Case 4: $x\in D,y\in D$ and $x\ne y$}

Here $x$ has no neighbor in $S$ but is dominated by $S'$, which occurs with probability $\mu \vartheta$; symmetrically, $y$ has no neighbor in $S'$ but is dominated by $S$, also with probability $\mu \vartheta$. The remaining $m-2$ vertices in $D$ are dominated by both $S$ and $S'$, each with probability $\zeta$. All vertices in $C$ are dominated by $S$, and all vertices in $B$ are dominated by $S'$, each with probability $1-\mu$. There are $m(m-1)$ ordered choices for $(x,y)$, yielding

$$P_4(i)=m(m-1)(\mu \vartheta)^2(1-\mu)^{2k-2i}\zeta^{m-2}.$$

Combining the above cases, we have
$$W(i)=P_1(i)+2P_2(i)+P_3(i)+P_4(i).$$

 \end{appendices}
 
\newpage


\end{document}